\documentclass[english,onecolumn]{IEEEtran}
\usepackage[T1]{fontenc}
\usepackage[latin9]{inputenc}
\usepackage{float}
\usepackage{bm}
\usepackage{amsmath}
\usepackage{graphicx}
\usepackage{setspace}
\usepackage{amssymb}
\makeatletter

\floatstyle{ruled}
\newfloat{algorithm}{tbp}{loa}
\floatname{algorithm}{Algorithm}

\newtheorem{definitn}{Definition}
\newtheorem{thm}{Theorem}
\newtheorem{remrk}{Remark}
\newtheorem{prop}{Proposition}

\usepackage{bm}
\usepackage{cite}
\usepackage{url}

\@ifundefined{showcaptionsetup}{}{%
 \PassOptionsToPackage{caption=false}{subfig}}
\usepackage{subfig}
\makeatother

\usepackage{babel}

\begin{document}

\title{Quantized Compressive Sensing }

\author{Wei Dai, Hoa Vinh Pham, and Olgica Milenkovic\\
 Department of Electrical and Computer Engineering\\
 University of Illinois at Urbana-Champaign}
\maketitle
\begin{abstract}
We study the average distortion introduced by scalar, vector, and
entropy coded quantization of compressive sensing (CS) measurements.
The asymptotic behavior of the underlying quantization schemes is
either quantified exactly or characterized via bounds. We adapt two
benchmark CS reconstruction algorithms to accommodate quantization
errors, and empirically demonstrate that these methods significantly
reduce the reconstruction distortion when compared to standard CS
techniques. 
\end{abstract}
\vspace{0.1in}

\renewcommand{\thefootnote}{\fnsymbol{footnote}} \footnotetext[1]{Part of the material in this paper was submitted to the IEEE Information Theory Workshop (ITW), 2009, and the IEEE International Symposium on Information Theory (ISIT), 2009.} \renewcommand{\thefootnote}{\arabic{footnote}} \setcounter{footnote}{0}

\section{\label{sec:Introduction}Introduction}

Compressive sensing (CS) is a linear sampling method that converts
unknown input signals, embedded in a high dimensional space, into
signals that lie in a space of significantly smaller dimension. In
general, it is not possible to uniquely recover an unknown signal
using measurements of reduced-dimensionality. Nevertheless, if the
input signal is sufficiently sparse, exact reconstruction is possible.
In this context, assume that the unknown signal $\mathbf{x}\in\mathbb{R}^{N}$
is $K$-sparse, i.e., that there are at most $K$ nonzero entries
in $\mathbf{x}$. A naive reconstruction method is to search among
all possible signals and find the sparsest one which is consistent
with the linear measurements. This method requires only $m=2K$ random
linear measurements, but finding the sparsest signal representation
is an NP-hard problem. On the other hand, Donoho and Candès et. al.
demonstrated in \cite{Donoho_IT2006_CompressedSensing,Candes_Tao_IT2005_decoding_linear_programming,Candes_Tao_FOCS05_Error_Correction_Linear_Programming,Candes_Tao_IT2006_Near_Optimal_Signal_Recovery}
that sparse signal reconstruction is a polynomial time problem if
more measurements are taken. This is achieved by casting the reconstruction
problem as a linear programming problem and solving it using the \emph{basis
pursuit (BP)} method. More recently, the authors proposed the \emph{subspace
pursuit (SP)} algorithm in \cite{Dai_2008_Subspace_Pursuit} (see
also the independent work \cite{Tropp2008_CoSamp} for a closely related
approach). The computational complexity of the SP algorithm is linear
in the signal dimension, and the required number of linear measurements
is of the same order as that for the BP method.

For most practical applications, it is reasonable to assume that the
measurements are quantized and therefore do not have infinite precision.
When the quantization error is bounded and known in advance, upper
bounds on the reconstruction distortion were derived for the BP method
in \cite{Candes_Tao_ApplMath2006_Stable_Signal_Recovery} and the
SP algorithm in \cite{Dai_2008_Subspace_Pursuit,Tropp2008_CoSamp},
respectively. For bounded compressible signals, which have transform
coefficients with magnitudes that decay according to a power law,
an upper bound on the reconstruction distortion introduced by a uniform
quantizer was derived in \cite{Candes2006_DCC_encoding_lp_ball}.
The same quantizer was studied in \cite{Baraniuk2008_quantization_sparse_representations}
for exactly $K$-sparse signals and it was shown that a large fraction
of quantization regions is not used \cite{Baraniuk2008_quantization_sparse_representations}.
All of the above approaches focus on the worst case analysis, or simple
one-bit quantization \cite{Braniuk2008_ciss_1bitquantization}. An
exception includes the overview paper \cite{Goyal2008_SPM_quantized_CS},
which focuses on the average performance of uniform quantizers, assuming
that the support set of the sparse signal is available at the quantizer.

As opposed to the worst case analysis, we consider the average distortion
introduced by quantization. We study the asymptotic distortion rate
functions for scalar quantization, entropy coded scalar quantization,
and vector quantization of the measurement vectors. Exact asymptotic
distortion rate functions are derived for scalar quantization when
both the measurement matrix and the sparse signals obey a certain
probabilistic model. Lower and upper bounds on the asymptotic distortion
rate functions are also derived for other quantization scenarios,
and the problem of compressive sensing matrix quantization is briefly
discussed as well. In addition, two benchmark CS reconstruction algorithms
are adapted to accommodate quantization errors. Simulations show that
the new algorithms offer significant performance improvement over
classical CS reconstruction techniques that do not take quantization
errors into consideration.

This paper is organized as follows. Section \ref{sec:Preliminaries}
contains a brief overview of CS theory, the BP and SP reconstruction
algorithms, and various quantization techniques. In Section \ref{sec:Distortion-Analysis},
we analyze the CS distortion rate function and examine the influence
of quantization errors on the BP and SP reconstruction algorithms.
In Section \ref{sec:Modified-Algorithm}, we describe two modifications
of the aforementioned algorithms, suitable for quantized data, that
offer significant performance improvements when compared to standard
BP and SP techniques. Simulation results are presented in Section
\ref{sec:Empirical-Results}.

\section{\label{sec:Preliminaries}Preliminaries}

\subsection{\label{sub:Compressive-Sensing}Compressive Sensing (CS)}

In CS, one encodes a signal $\mathbf{x}$ of dimension $N$ by computing
a measurement vector $\mathbf{y}$ of dimension of $m\ll N$ via linear
projections, i.e., \[
\mathbf{y}=\mathbf{\Phi}\mathbf{x},\]
 where $\mathbf{\Phi}\in\mathbb{R}^{m\times N}$ is referred to as
the \emph{measurement matrix}. In this paper, we assume that $\mathbf{x}\in\mathbb{R}^{N}$
is exactly $K$-sparse, i.e., that there are exactly $K$ entries
of $\mathbf{x}$ that are nonzero. The reconstruction problem is to
recover $\mathbf{x}$ given $\mathbf{y}$ and $\mathbf{\Phi}$.

The BP method is a technique that casts the reconstruction problem
as a $l_{1}$-regularized optimization problem, i.e., \begin{equation}
\min\;\left\Vert \mathbf{x}\right\Vert _{1}\;\mathrm{subject\; to}\;\mathbf{y}=\mathbf{\Phi}\mathbf{x},\label{eq:BP-standard}\end{equation}
 where $\left\Vert \mathbf{x}\right\Vert _{1}=\sum_{i=1}^{N}\left|x_{i}\right|$
denotes the $l_{1}$-norm of the vector $\mathbf{x}$. It is a convex
optimization problem and can be solved efficiently by linear programming
techniques. The reconstruction complexity equals $O\left(m^{2}N^{3/2}\right)$
if the convex optimization problem is solved using interior point
methods \cite{Nesterov_book1994_Interior_point_Convex_Programming}.

The computational complexity of CS reconstruction can be further reduced
by the SP algorithm, recently proposed by two research groups \cite{Dai_2008_Subspace_Pursuit,Tropp2008_CoSamp}.
It is an iterative algorithm drawing on the theory of list decoding.
The computational complexity of this algorithm is upper bounded by
$O\left(Km(N+K^{2})\right)$, which is significantly smaller than
the complexity of the BP method whenever $K\ll N$. See \cite{Dai_2008_Subspace_Pursuit}
for a detailed performance and complexity analysis of this greedy
algorithm.

A sufficient condition for both the BP and SP algorithms to perform
exact reconstruction is based on the so called restricted isometry
property (RIP) \cite{Candes_Tao_IT2005_decoding_linear_programming},
formally defined as follows.
\begin{definitn}
\label{def:RIP}\emph{(RIP).} A matrix $\mathbf{\Phi}\in\mathbb{R}^{m\times N}$
is said to satisfy the Restricted Isometry Property (RIP) with coefficients
$\left(K,\delta\right)$ for $K\le m$, $0\leq\delta\leq1$, if for
all index sets $I\subset\left\{ 1,\cdots,N\right\} $ such that $\left|I\right|\le K$
and for all $\mathbf{q}\in\mathbb{R}^{\left|I\right|}$, one has \[
\left(1-\delta\right)\left\Vert \mathbf{q}\right\Vert _{2}^{2}\le\left\Vert \mathbf{\Phi}_{I}\mathbf{q}\right\Vert _{2}^{2}\le\left(1+\delta\right)\left\Vert \mathbf{q}\right\Vert _{2}^{2}.\]

\end{definitn}
The RIP parameter $\delta_{K}$ is defined as the infimum of all parameters
$\delta$ for which the RIP holds, i.e., \begin{align}
\delta_{K} & :=\inf\left\{ \delta:\;\left(1-\delta\right)\left\Vert \mathbf{q}\right\Vert _{2}^{2}\le\left\Vert \mathbf{\Phi}_{I}\mathbf{q}\right\Vert _{2}^{2}\le\left(1+\delta\right)\left\Vert \mathbf{q}\right\Vert _{2}^{2},\right.\nonumber \\
 & \quad\quad\quad\quad\left.\;\forall\left|I\right|\le K,\;\forall\mathbf{q}\in\mathbb{R}^{\left|I\right|}\right\} .\label{eq:def-RIP-parameter}\end{align}

\vspace{0.02in}

It was shown in \cite{Candes_Tao_ApplMath2006_Stable_Signal_Recovery,Dai_2008_Subspace_Pursuit}
that both BP and SP algorithms lead to exact reconstructions of $K$-sparse
signals if the matrix $\mathbf{\Phi}$ satisfies the RIP with a constant
parameter, i.e., $\delta_{c_{1}K}\le c_{0}$ where both $c_{1}\in\mathbb{R}^{+}$
and $c_{0}\in\left(0,1\right)$ are constants independent of $K$
(although different algorithms may have different parameters $c_{0}$s
and $c_{1}$s). Most known families of matrices satisfying the RIP
property with optimal or near-optimal performance guarantees are random,
including Gaussian random matrices with i.i.d. $\mathcal{N}\left(0,1/m\right)$
entries, where $m\ge O\left(K\log N\right)$.

For completeness, we briefly describe the SP algorithm. For an index
set $T\subset\left\{ 1,2,\cdots,N\right\} $, let $\mathbf{\Phi}_{T}$
be the {}``truncated matrix'' consisting of the columns of $\mathbf{\Phi}$
indexed by $T$, and let $\mathrm{span}\left(\mathbf{\Phi}_{T}\right)$
denote the subspace in $\mathbb{R}^{m}$ spanned by the columns of
$\mathbf{\Phi}_{T}$. Suppose that $\mathbf{\Phi}_{T}^{*}\mathbf{\Phi}_{T}$
is invertible. For any given $\mathbf{y}\in\mathbb{R}^{m}$, the projection
of $\mathbf{y}$ onto $\mathrm{span}\left(\mathbf{\Phi}_{T}\right)$
is defined as\begin{align}
 & \mathbf{y}_{p}=\mathrm{proj}\left(\mathbf{y},\mathbf{\Phi}_{T}\right):=\mathbf{\Phi}_{T}\left(\mathbf{\Phi}_{T}^{*}\mathbf{\Phi}_{T}\right)^{-1}\mathbf{\Phi}_{T}^{*}\mathbf{y},\label{eq:def-proj}\end{align}
 where $\mathbf{\Phi}^{*}$ denotes the conjugate transpose of $\mathbf{\Phi}$.

The corresponding projection residue vector $\mathbf{y}_{r}$ and
projection coefficient vector $\mathbf{x}_{p}$ are defined as \begin{equation}
\mathbf{y}_{r}=\mathrm{resid}\left(\mathbf{y},\mathbf{\Phi}_{T}\right):=\mathbf{y}-\mathbf{y}_{p},\label{eq:def-proj-residue}\end{equation}
 and\begin{align}
 & \mathbf{x}_{p}=\mathrm{pcoeff}\left(\mathbf{y},\mathbf{\Phi}_{T}\right):=\left(\mathbf{\Phi}_{T}^{*}\mathbf{\Phi}_{T}\right)^{-1}\mathbf{\Phi}_{T}^{*}\mathbf{y}.\label{eq:def-proj-coeff}\end{align}
 The steps of the SP algorithm are summarized below.

\begin{algorithm}[H]
 \textbf{Input}: $K$, $\mathbf{\Phi}$, $\mathbf{y}$

\noindent \textbf{Initialization}: Let $T^{0}=\left\{ K^{\phantom{*}}\right.$indices
corresponding to entries of largest magnitude in $\left.\mathbf{\Phi}^{*}\mathbf{y}\right\} $
and $\mathbf{y}_{r}^{0}=\mathrm{resid}\left(\mathbf{y},\mathbf{\Phi}_{\hat{T}^{0}}\right)$.

\noindent \textbf{Iteration}: At the $\ell^{\mathrm{th}}$ iteration,
go through the following steps. 
\begin{enumerate}
\item $\tilde{T}^{\ell}=T^{\ell-1}\bigcup$$\left\{ K^{\phantom{*}}\right.$indices
corresponding to entries of largest magnitude in $\left.\mathbf{\Phi}^{*}\mathbf{y}_{r}^{\ell-1}\right\} $. 
\item Let $\mathbf{x}_{p}=\mathrm{pcoeff}\left(\mathbf{y},\mathbf{\Phi}_{\tilde{T}^{\ell}}\right)$
and $T^{\ell}=\left\{ K^{\phantom{*}}\right.$indices corresponding
to entries of largest magnitude in $\left.\mathbf{x}_{p}\right\} $. 
\item $\mathbf{y}_{r}^{\ell}=\mathrm{resid}\left(\mathbf{y},\mathbf{\Phi}_{T^{\ell}}\right).$ 
\item If $\left\Vert \mathbf{y}_{r}^{\ell}\right\Vert _{2}>\left\Vert \mathbf{y}_{r}^{\ell-1}\right\Vert _{2}$,
let $T^{\ell}=T^{\ell-1}$ and quit the iteration. 
\end{enumerate}
\textbf{Output}: The vector $\hat{\mathbf{x}}$ satisfying $\hat{\mathbf{x}}_{\left\{ 1,\cdots,N\right\} -T^{\ell}}=\mathbf{0}$
and $\hat{\mathbf{x}}_{T^{\ell}}=\mathrm{pcoeff}\left(\mathbf{y},\mathbf{\Phi}_{T^{\ell}}\right)$.

\caption{\label{alg:Subspace-Pursuit-Algorithm}The Subspace Pursuit (SP) Algorithm}

\end{algorithm}

In what follows, we study the performance of the SP and BP reconstruction
algorithms when the measurements are subjected to three different
quantization schemes. We also discuss the issue of quantizing the
measurement matrix values.

\subsection{\label{sub:Scalar-Quantization}Scalar and Vector Quantization}

Let $\mathcal{C}\subset\mathbb{R}^{m}$ be a finite discrete set,
referred to as a codebook. A quantizer is a mapping from $\mathbb{R}^{m}$
to the codebook $\mathcal{C}$ with the property that \begin{align}
\mathfrak{q}:\;\mathbb{R}^{m} & \rightarrow\mathcal{C}\nonumber \\
\mathbf{y} & \mapsto\bm{\omega}\in\mathcal{C}\;\mathrm{if}\;\mathbf{y}\in\mathcal{R}_{\bm{\omega}},\label{eq:scalar-quant}\end{align}
 where $\bm{\omega}$ is referred to as a \emph{level} and $\mathcal{R}_{\bm{\omega}}$
is the \emph{quantization region} corresponding to the level $\bm{\omega}$.
The performance of a quantizer is often described by its distortion-rate
function, defined as follows. Let the distortion measure be the squared
Euclidean distance (i.e., mean squared error (MSE)). For a random
source $\mathbf{Y}\in\mathbb{R}^{m}$, the distortion associated with
a quantizer $\mathfrak{q}$ is $D_{\mathfrak{q}}:=\mathrm{E}\left[\left\Vert \mathbf{Y}-\mathfrak{q}\left(\mathbf{Y}\right)\right\Vert _{2}^{2}\right]$.
For a given codebook $\mathcal{C}$, the optimal quantization function
that minimizes the Euclidean distortion measure is given by \[
\mathfrak{q}^{*}\left(\mathbf{Y}\right)=\underset{\bm{\omega}\in\mathcal{C}}{\arg\;\min}\;\left\Vert \mathbf{Y}-\bm{\omega}\right\Vert _{2}^{2}.\]
 As a result, the corresponding quantization region is given by \begin{equation}
\mathcal{R}_{\bm{\omega}}:=\left\{ \mathbf{y}\in\mathbb{R}^{m}:\;\left\Vert \mathbf{y}-\bm{\omega}\right\Vert _{2}^{2}\le\left\Vert \mathbf{y}-\bm{\omega}^{\prime}\right\Vert _{2}^{2},\;\forall\bm{\omega}^{\prime}\in\mathcal{C}\right\} ,\label{eq:def-quant-region}\end{equation}
 and the distortion associated with this codebook $\mathcal{C}$ equals
\[
D\left(\mathcal{C}\right):=\mathrm{E}\left[\left\Vert \mathbf{Y}-\mathfrak{q}^{*}\left(\mathbf{Y}\right)\right\Vert _{2}^{2}\right].\]
 Let $R:=\frac{1}{m}\log_{2}\left|\mathcal{C}\right|$ be the rate
of the codebook $\mathcal{C}$. For a given code rate $R$, the distortion
rate function is given by

\begin{equation}
D^{*}\left(R\right):=\underset{\mathcal{C}:\;\frac{1}{m}\log_{2}\left|\mathcal{C}\right|\le R}{\inf}\; D\left(\mathcal{C}\right).\label{eq:def-DRF-vector}\end{equation}
 For simplicity, assume that the random source $\mathbf{Y}$ does
not have mass points, and that the levels in the quantization codebook
are all distinct. With these assumptions, though different quantization
regions (\ref{eq:def-quant-region}) may overlap, the ties can be
broken arbitrarily as they happen with probability zero.

We study both vector quantization and scalar quantization. Scalar
quantization has lower computational complexity than vector quantization.
It is a special case of vector quantization when $m=1$. To distinguish
the two schemes, we use the subscripts $SQ$ and $VQ$ to refer to
scalar and vector quantization, respectively. For quantized compressive
sensing, we assume that the quantization functions for all the coordinate
of $\mathbf{Y}$ are the same. The corresponding distortion rate function
is therefore of the form \begin{equation}
D_{SQ}^{*}\left(R\right):=\underset{\mathcal{C}_{SQ}:\;\log_{2}\left|\mathcal{C}_{SQ}\right|\le R}{\inf}\mathrm{E}_{\mathbf{Y}}\left[\sum_{i=1}^{m}\left|Y_{i}-\mathfrak{q}_{SQ}\left(Y_{i}\right)\right|^{2}\right].\label{eq:def-DRF-scalar}\end{equation}

Necessary conditions for optimal scalar quantizer design can be found
in \cite{Lloyd1982_IT_quantization}. The quantization region for
the level $\omega_{i}\in\mathcal{C}$, $i=1,2,\cdots,2^{R}$, can
be written in the form $\mathcal{R}_{\omega_{i}}=\overline{\left(t_{i-1},t_{i}\right)}$,
where $t_{i-1},t_{i}\in\mathbb{R}\bigcup\left\{ -\infty\right\} \bigcup\left\{ +\infty\right\} $
and $\overline{\left(t_{i-1},t_{i}\right)}$ is the closure of the
open interval $\left(t_{i-1},t_{i}\right)$. An optimal quantizer
satisfies the following conditions: 
\begin{enumerate}
\item If the optimal quantizer has levels $\omega_{i-1}$ and $\omega_{i}$,
then the threshold that minimizes the mean square error (MSE) is \begin{equation}
t_{i}=\frac{1}{2}\left(\omega_{i}+\omega_{i+1}\right).\label{eq:thresh-update}\end{equation}

\item If the optimal quantizer has thresholds $t_{i-1}$ and $t_{i}$, then
the level that minimizes the MSE is \begin{equation}
\omega_{i}=\mathrm{E}\left[Y|Y\in\overline{\left(t_{i-1},t_{i}\right)}\right].\label{eq:center-update}\end{equation}

\end{enumerate}
\vspace{0.02in}

Lloyd's algorithm \cite{Lloyd1982_IT_quantization} for quantizer
codebook design is based on the above necessary conditions. Lloyd's
algorithm starts with an initial codebook, and then in each iteration,
computes the thresholds $t_{i}$s according to (\ref{eq:thresh-update})
and updates the codebook via (\ref{eq:center-update}). Although Lloyd's
algorithm is not guaranteed to find a global optimum for the quantization
regions, it produces locally optimal codebooks.

As a low-complexity alternative to non-uniform quantizers, uniform
scalar quantizers are widely used in practice. A uniform scalar quantizer
is associated with a {}``uniform codebook'' $\mathcal{C}_{u,SQ}=\left\{ \omega_{1}<\omega_{2}<\cdots<\omega_{M}\right\} ,$
for which $\omega_{i}-\omega_{i-1}=\omega_{j}-\omega_{j-1}$ for all
$1<i\ne j\le M$. The difference between adjacent levels is often
referred to as the step size, and denoted by $\Delta_{u,SQ}$. The
corresponding distortion rate function is given by \begin{align}
D_{u,SQ}^{*}\left(R\right) & :=\underset{\mathcal{C}_{u,SQ}:\;\log_{2}\left|\mathcal{C}_{u,SQ}\right|\le R}{\inf}\nonumber \\
 & \qquad\quad\mathrm{E}_{\mathbf{Y}}\left[\sum_{i=1}^{m}\left|Y_{i}-\mathfrak{q}_{SQ}\left(Y_{i}\right)\right|^{2}\right].\label{eq:def-DRF-uniform-scalar}\end{align}
 where $\mathcal{C}_{SQ}$ in (\ref{eq:def-DRF-scalar}) is replaced
by $\mathcal{C}_{u,SQ}$.

Definitions (\ref{eq:def-DRF-scalar}) and (\ref{eq:def-DRF-uniform-scalar})
are consistent with (\ref{eq:def-DRF-vector}) as a Cartesian product
of scalar quantizers can be viewed as a special form of a vector quantizer.

\section{\label{sec:Distortion-Analysis}Distortion Analysis}

We analyze the asymptotic behavior of the distortion rate functions
introduced in the previous section. We assume that the quantization
codebook $\mathcal{C}$, for both scalar and vector quantization,
is designed offline and fixed when the measurements are taken.

\subsection{Distortion of Scalar Quantization}

For scalar quantization, we consider the following two CS scenarios.

\begin{flushleft}
\textbf{Assumptions I}: 
\par\end{flushleft}
\begin{enumerate}
\item Let $\mathbf{\Phi}=\frac{1}{\sqrt{m}}\mathbf{A}\in\mathbb{R}^{m\times N}$,
where the entries of $\mathbf{A}$ are i.i.d. Subgaussian random variables%
\footnote{A random variable $X$ is said to be \emph{Subgaussian} if there exist
positive constants $c_{1}$ and $c_{2}$ such that \[
\Pr\left(\left|X\right|>x\right)\le c_{1}e^{-c_{2}x^{2}}\quad\forall x>0.\]

One property of Subgaussian distributions is that they have a well
defined moment generating function. Note that the Gaussian and Bernoulli
distributions are special cases of the Subgaussian distribution. %
} with zero mean and unit variance. 
\item Let $\mathbf{X}\in\mathbb{R}^{N}$ be an exactly $K$-sparse vector,
that is, a signal that has exactly $K$ nonzero entries. We assume
that the nonzero entries of $\mathbf{X}$ are i.i.d. Subgaussian random
variables with zero mean and unit variance, although more general
models can be analyzed in a similar manner. 
\end{enumerate}
\vspace{0.01in}

\begin{flushleft}
\textbf{Assumptions II}: Assume that $\mathbf{X}\in\mathbb{R}^{n}$
is exactly $K$-sparse, and that the nonzero entries of $\mathbf{X}$
are i.i.d. standard Gaussian random variables. 
\par\end{flushleft}

\vspace{0.01in}

The asymptotic distortion-rate function of the measurement vector
under the first CS scenario is characterized in Theorem \ref{thm:DRF-Gaussian-Matrix}. 
\begin{thm}
\label{thm:DRF-Gaussian-Matrix}Suppose that Assumptions I hold. Then
\begin{equation}
\underset{R\rightarrow\infty}{\lim}\underset{\left(K,m,N\right)\rightarrow\infty}{\lim}\;\frac{2^{2R}}{K}D_{SQ}^{*}\left(R\right)=\frac{\pi\sqrt{3}}{2},\label{eq:DRF-Gaussian-non-uniform}\end{equation}
 and \begin{equation}
\underset{R\rightarrow\infty}{\lim}\underset{\left(K,m,N\right)\rightarrow\infty}{\lim}\;\frac{2^{2R}}{KR}D_{u,SQ}^{*}\left(R\right)=\frac{4}{3}\ln2.\label{eq:DRF-Gaussian-uniform}\end{equation}

\end{thm}
\vspace{0.01in}

The proof is based on the fact that the distributions of $\sqrt{\frac{m}{K}}Y_{i}$,
$1\le i\le m$, weakly converge to standard Gaussian distributions.
The detailed description is given in Appendix \ref{sub:Pf-Thm-Gaussian-scalar}.

To study the scenario described by Assumptions II, we need the following
definitions. For a given matrix $\mathbf{\Phi}$, let \begin{equation}
\mu_{1}:=\frac{1}{N}\sum_{i\in\left[m\right],j\in\left[N\right]}\varphi_{i,j}^{2},\label{eq:def-mu-1}\end{equation}
 and \begin{equation}
\mu_{2}:=\underset{i\in\left[m\right],T\in{\left[N\right] \choose K}}{\max}\frac{m}{K}\sum_{j\in T}\varphi_{i,j}^{2},\label{eq:def-mu-2}\end{equation}
 where $\left[m\right]=\left\{ 1,2,\cdots,m\right\} $ and ${\left[N\right] \choose K}$
denotes the set of all subsets of $\left[N\right]$ with cardinality
$K$. Note that if the matrix $\mathbf{\Phi}$ is generated from the
random ensemble described in Assumption I.1), then $\mu_{1}\in\left(1-\epsilon,1+\epsilon\right)$
with high probability, for all $\epsilon>0$, and whenever $m$ and
$N$ are sufficiently large. It is straightforward to verify that
$\mu_{1}\le\mu_{2}$.

With these definitions at hand, bounds on the distortion rate function
can be described as below. 
\begin{thm}
\label{thm:DRF-Matrix}Suppose that Assumption II holds. Then\begin{align}
 & \frac{\pi\sqrt{3}}{2}\mu_{1}\le\underset{R\rightarrow\infty}{\lim\inf}\frac{2^{2R}}{K}D_{SQ}^{*}\left(R\right)\nonumber \\
 & \quad\le\underset{R\rightarrow\infty}{\lim\sup}\frac{2^{2R}}{K}D_{SQ}^{*}\left(R\right)\le\frac{\pi\sqrt{3}}{2}\mu_{2},\label{eq:DRF-non-uniform}\end{align}
 and \begin{equation}
\frac{4\ln2}{3}\mu_{1}\le\underset{R\rightarrow\infty}{\lim\inf}\frac{2^{2R}}{KR}D_{u,SQ}^{*}\left(R\right).\label{eq:DRF-uniform}\end{equation}

\end{thm}
\vspace{0.01in}

The detailed proof is postponed to Appendix \ref{sub:Pf-lb-ub}. Here,
we sketch the basic ideas behind the proof. In order to construct
a lower bound, suppose that one has prior information about the support
set $T$ before taking the measurements. For a given value of $i$
and for a given $T$, we calculate the corresponding asymptotic distortion-rate
function. The lower bound is obtained by taking the average of these
distortion-rate functions over all possible values of $i$ and $T$.
For the upper bound, we design a sequence of sub-optimal scalar quantizers,
then apply them to all measurement components, and finally construct
a uniform upper bound on their asymptotic distortion-rate functions,
valid for all $i$ and $T$. The uniform upper bound is given in (\ref{eq:DRF-non-uniform}). 
\begin{remrk}
Our results are based on the fundamental assumption that the sparsity
level $K$ is known in advance and that the statistics of the sparse
vector $\mathbf{x}$ is specified. Very frequently, however, this
is not the case in practice. If we relax Assumptions I and II further
by assuming that $K$ is sufficiently large, it will often be the
case that the statistics of the measurement $Y_{i}$ is well approximated
by a Gaussian distribution. Here, note that different $Y_{i}$ variables
may have different variances and these variances are generally unknown
in advance. The problem of statistical mismatch has been analyzed
in the proof of the upper bound (\ref{eq:DRF-non-uniform}) (see Proposition
\ref{pro:ub-mismatch} of Appendix \ref{sub:Pf-lb-ub} for details).
In particular, non-uniform quantization with slightly over-estimated
variance performs better than that with under-estimated variance \cite[Chapter 8.6]{Sayood2005_book_data_compression}.
 \vspace{0.01in}

\end{remrk}
According to Theorem \ref{thm:DRF-Gaussian-Matrix}, if the quantization
rate $R$ is sufficiently large, the distortion of the optimal non-uniform
quantizer is approximately only $1/R$ of that of the optimal uniform
quantizer. This gap can be closed by using entropy coding techniques
in conjunction with uniform quantizers.

\subsection{Uniform Scalar Quantization with Entropy Encoding}

Let $\mathcal{B}_{enc}=\left\{ \bm{v}_{1},\bm{v}_{2},\cdots,\bm{v}_{M}\right\} $
be a binary codebook, where the codewords $\bm{v}_{i}$, $1\le i\le M$,
are finite-length strings over the binary field with elements $\left\{ 0,1\right\} $.
The codebook $\mathcal{B}_{enc}$ can, in general, contain codewords
of variable length - i.e., the lengths of different codewords are
allowed to be different. Let $\ell_{i}$ be the length of codeword
$\bm{v}_{i}$, $i=1,2,\cdots,M$. Then $\bm{v}_{i}\in\left\{ 0,1\right\} ^{\ell_{i}\times1}$.
For a given quantization codebook $\mathcal{C}=\left\{ \omega_{1},\omega_{2},\cdots,\omega_{M}\right\} $,
the encoding function $\mathfrak{f}_{enc}$ is a mapping from the
quantization codebook $\mathcal{C}$ to the binary codebook $\mathcal{B}_{enc}$,
i.e., $\mathfrak{f}_{enc}\left(\omega\right)=\bm{v}\in\mathcal{B}_{enc}$.
The extension $\mathfrak{f}_{enc}^{*}$ is a mapping from finite length
strings of $\mathcal{C}$ to finite length strings of $\mathcal{B}_{enc}$
(a concatenation of the corresponding binary codewords): \[
\mathfrak{f}_{enc}^{*}\left(\omega_{i_{1}}\omega_{i_{2}}\cdots\omega_{i_{s}}\right)=\mathfrak{f}_{enc}\left(\omega_{i_{1}}\right)\mathfrak{f}_{enc}\left(\omega_{i_{2}}\right)\cdots\mathfrak{f}_{enc}\left(\omega_{i_{s}}\right).\]
 The code $\mathcal{B}_{enc}$ is called \emph{uniquely decodable}
if any concatenation of binary codewords $\bm{v}_{i_{1}}\bm{v}_{i_{2}}\cdots\bm{v}_{i_{s}}$
has only one possible preimage string $\omega_{j_{1}}\omega_{j_{2}}\cdots\omega_{j_{s}}$
producing it. In practice, the code $\mathcal{B}_{enc}$ is often
chosen to be a \emph{prefix} code, that is, no codeword is a prefix
of any other codeword. A prefix code can be uniquely decoded as the
end of a codeword is immediately recognizable without checking future
encoded bits.

We consider the case in which scalar quantization is followed by variable-length
encoding. The corresponding expected encoding length $\bar{L}$ is
defined by \[
\bar{L}=\mathrm{E}_{Y}\left[\mathfrak{L}\circ\mathfrak{f}_{enc}\circ\mathfrak{q}_{SQ}\left(Y\right)\right],\]
 where $\mathfrak{L}\left(\bm{v}\right)$ outputs the length of the
encoding codeword $\bm{v}\in\mathcal{B}_{enc}$. The goal is to \emph{jointly}
design $\mathfrak{q}_{SQ}$ and $\mathfrak{f}_{enc}$ to minimize
the expected encoding length $\bar{L}$. We are interested in the
distortion rate function defined by \begin{equation}
D_{enc}^{*}\left(R\right):=\underset{\bar{L}\le R}{\inf}\;\mathrm{E}_{\mathbf{Y}}\left[\sum_{i=1}^{m}\left|Y_{i}-\mathfrak{q}_{SQ}\left(Y_{i}\right)\right|^{2}\right].\label{eq:DRF-encoding}\end{equation}

\begin{thm}
\label{thm:Encoding}Suppose that Assumptions I hold. Then\begin{align*}
\frac{\pi e}{6} & \le\underset{R\rightarrow\infty}{\lim\inf}\underset{\left(K,m,N\right)\rightarrow\infty}{\lim\inf}\frac{2^{2R}}{K}D_{enc}^{*}\left(R\right)\\
 & \le\underset{R\rightarrow\infty}{\lim\sup}\underset{\left(K,m,N\right)\rightarrow\infty}{\lim\sup}\frac{2^{2R}}{K}D_{enc}^{*}\left(R\right)\le\frac{\pi e}{3},\end{align*}
 and the upper bound is achieved by a uniform scalar quantizer with
\[
\underset{R\rightarrow\infty}{\lim}\underset{\left(K,m,N\right)\rightarrow\infty}{\lim}\sqrt{\frac{m}{2\pi eK}}2^{R}\Delta_{u,SQ}=1,\]
 followed by Huffmann encoding. 
\end{thm}
\vspace{0.01in}

\begin{proof}
Given a quantization function, Huffmann encoding gives an optimal
prefix code that minimizes $\bar{L}$ \cite[Chapter 5]{Cover1991_book_information_theory}.
Let $p_{i}=\Pr\left(Y:\;\mathfrak{q}\left(Y\right)=\omega_{i}\right)$
and let $\ell_{i}$ be the length of encoded codeword $\mathfrak{f}_{enc}\left(\omega_{i}\right)$.
Let $H:=\sum_{i=1}^{M}-p_{i}\log_{2}p_{i}$. Then $H\le\bar{L}=\sum_{i}p_{i}\ell_{i}\le H+1$.
In addition, it is well known that the distortion of scalar quantization
of a Gaussian source is lower bounded by $\frac{1}{12}2^{2\left(h-H\right)}\left(1+o_{H}\left(1\right)\right)$,
where $h$ denotes the differential entropy of the source, and the
lower bound is achieved by a uniform quantizer. Calculating $h$ and
interpreting $H$ as a function of $\bar{L}$ establish the claimed
result. 
\end{proof}
\vspace{0.01in}

As expected, for a given average description length, the average distortion
of uniform scalar quantization and Huffmann encoding is smaller than
that of an optimal scalar quantizer with fixed length encoding.

\subsection{Distortion of Vector Quantization}

For the purpose of analyzing vector quantization schemes, we make
the following assumptions.

\begin{flushleft}
\textbf{Assumptions III}: 
\par\end{flushleft}
\begin{enumerate}
\item Let $\mathbf{\Phi}\in\mathbb{R}^{m\times N}$ be a matrix satisfying
the RIP with parameter $\delta_{K}\in\left(0,1\right)$. 
\item Assume that $\mathbf{X}\in\mathbb{R}^{n}$ is exactly $K$-sparse,
and that the nonzero entries of $\mathbf{X}$ are i.i.d. standard
Gaussian random variables. 
\end{enumerate}
\vspace{0.01in}

\begin{thm}
\label{thm:DRF-vector}Suppose that Assumptions III hold. Then\begin{align}
 & \left(1-\delta_{K}\right)\left(1+o_{K}\left(1\right)\right)\le\underset{R\rightarrow\infty}{\lim\inf}\frac{2^{2Rm/K}}{K}D_{VQ}^{*}\left(R\right)\label{eq:lb-DRF-VQ}\\
 & \quad\le\underset{R\rightarrow\infty}{\lim\sup}\frac{2^{2R}}{m}D_{VQ}^{*}\left(R\right)\le\left(1+\delta_{K}\right)\left(1+o_{m}\left(1\right)\right),\label{eq:ub-DRF-VQ-1}\end{align}
 where $o_{K}\left(1\right)\overset{K\rightarrow\infty}{\rightarrow}0$
and $o_{m}\left(1\right)\overset{m\rightarrow\infty}{\rightarrow}0$.
Another upper bound on $D_{VQ}^{*}\left(R\right)$ is given by \begin{equation}
\underset{R\rightarrow\infty}{\lim\sup}\frac{2^{2R}}{K}D_{VQ}^{*}\left(R\right)\le\frac{\pi\sqrt{3}}{2}\mu_{2},\label{eq:ub-DRF-VQ-2}\end{equation}
 where $\mu_{2}$ is as defined in (\ref{eq:def-mu-2}). 
\end{thm}
\vspace{0.01in}

\begin{remrk}
The comparison of the two upper bounds in (\ref{eq:ub-DRF-VQ-1})
and (\ref{eq:ub-DRF-VQ-2}) depends on the ratio between $m$ and
$K$. Consider the case where $N=\beta K$, $m=\Theta\left(K\log\left(N/K\right)\right)=\alpha K$
for some $\alpha,\beta>1$. The first upper bound becomes \[
\underset{R\rightarrow\infty}{\lim\sup}\frac{2^{2R}}{K}D_{VQ}^{*}\left(R\right)\le\alpha\left(1+\delta_{K}\right)\left(1+o_{m}\left(1\right)\right).\]
It is smaller than the second upper bound if and only if \[
\delta_{K}<\frac{\pi\sqrt{3}}{2\alpha}\mu_{2}-1.\]

\vspace{0.01in}

\end{remrk}
The upper bound (\ref{eq:ub-DRF-VQ-2}) is obtained by using the Cartesian
product of scalar quantizers and invoking the result in (\ref{eq:DRF-non-uniform}).
The bounds (\ref{eq:lb-DRF-VQ}) and (\ref{eq:ub-DRF-VQ-1}) are proved
in Appendix \ref{sub:Pf-lb-ub}. The basic ideas behind the proof
are similar to those used for proving Theorem \ref{thm:DRF-Matrix}:
the lower bound is obtained by averaging the distortions of optimal
quantizers for every $T\in{\left[N\right] \choose K}$, while the
upper bound is a uniform upper bound on the distortions of quantizers
constructed for all $T\in{\left[N\right] \choose K}$.

Note that the lower bound in (\ref{eq:lb-DRF-VQ}) is not achievable
when $K<m$. The upper bounds (\ref{eq:ub-DRF-VQ-1}) and (\ref{eq:ub-DRF-VQ-2})
do not guarantee significant distortion reduction of vector quantization
compared with scalar quantization. Due to their inherently high computational
complexity, vector quantizers do not offer clear advantages that justify
their use in practice.

\subsection{CS Measurement Matrix Quantization Effects}

In CS theory, the measurement matrix is generated either randomly
or by some deterministic construction. Examples include Gaussian random
matrices and the deterministic construction based on Vandermonde matrices
\cite{Tarokh2007_ISIT_CS_ReedSolomon,Shokrollahi2009_bit_precision_CS}.
In both examples, the matrix entries typically have infinite precision,
which is not the case in practice. It is therefore also plausible
to study the effect of quantization of CS measurement matrix. 

Consider Assumption I where the measurement matrix is randomly generated.
Let us assume that every entry $\varphi_{i,j}$, $1\le i\le m$ and
$1\le j\le N$, is quantized using a finite number of bits. Note that
$\hat{\varphi}_{i,j}=\mathfrak{q}\left(\varphi_{i,j}\right)$ is a
bounded random variable and therefore Subgaussian distributed. The
results in Theorem \ref{thm:DRF-Gaussian-Matrix} are therefore automatically
valid for quantized matrices as well. 

Suppose that the measurement matrix is constructed deterministically
and then quantized using a finite number of bits. The parameters $\mu_{1}$,
$\mu_{2}$ and $\delta_{K}$ of the quantized measurement matrix can
be computed according to (\ref{eq:def-mu-1}), (\ref{eq:def-mu-2})
and (\ref{eq:def-RIP-parameter}), respectively. The results regarding
scalar quantization and vector quantization described in Theorems
\ref{thm:DRF-Matrix} and \ref{thm:DRF-vector} can be easily seen
to hold in this case as well.

\subsection{Reconstruction Distortion}

Based on the results of the previous section, we are ready to quantify
the reconstruction distortion of BP and SP methods introduced by quantization
error.

It is well known from CS literature that the reconstruction distortion
is dependent on the distortion in the measurements. Consider the quantized
CS given by \[
\hat{\mathbf{Y}}=\mathfrak{q}\left(\mathbf{Y}\right)=\mathbf{\Phi}\mathbf{X}+\mathbf{E},\]
 and where $\mathbf{E}\in\mathbb{R}^{m}$ denotes the quantization
error. Let $\hat{\mathbf{X}}$ be the reconstructed signal based on
the quantized measurements $\hat{\mathbf{Y}}$. Then the reconstruction
distortion can be upper bounded by \begin{equation}
\left\Vert \mathbf{X}-\hat{\mathbf{X}}\right\Vert _{2}^{2}\le c^{2}\left\Vert \mathbf{E}\right\Vert _{2}^{2},\label{eq:ub-reconst-dist}\end{equation}
 where the constant $c$ differs for different reconstruction algorithms.
The best bounding constant for the BP method was given in \cite{Candes_Tao_ApplMath2006_Stable_Signal_Recovery},
and equals \[
c_{bp}=\frac{4}{\sqrt{3-3\delta_{4K}}-\sqrt{1+\delta_{4K}}},\]
 while for the SP algorithm, the constant was estimated in \cite{Dai_2008_Subspace_Pursuit}
\[
c_{sp}=\frac{1+\delta_{3K}+\delta_{3K}^{2}}{\delta_{3K}\left(1-\delta_{3K}\right)}.\]

A lower bound on the reconstruction distortion is given as follows.
Suppose that the support set $T$ of the sparse signal $\mathbf{x}$
is perfectly reconstructed. The reconstructed signal $\hat{\mathbf{X}}$
is given by \[
\hat{\mathbf{X}}=\left(\mathbf{\Phi}_{T}^{*}\mathbf{\Phi}_{T}\right)^{-1}\mathbf{\Phi}_{T}\hat{\mathbf{Y}},\]
and the reconstruction distortion is lower bounded by \begin{equation}
\left\Vert \hat{\mathbf{X}}-\mathbf{X}\right\Vert _{2}^{2}\ge\left(\frac{\sqrt{1-\delta_{K}}}{1+\delta_{K}}\right)^{2}\left\Vert \hat{\mathbf{Y}}-\mathbf{Y}\right\Vert _{2}^{2}=\frac{1-\delta_{K}}{\left(1+\delta_{K}\right)^{2}}\left\Vert \mathbf{E}\right\Vert _{2}^{2}.\label{eq:lb-reconst-dist}\end{equation}
For short, let \[
c_{lb}=\frac{\sqrt{1-\delta_{K}}}{1+\delta_{K}}.\]

Combining the bounds (\ref{eq:ub-reconst-dist},\ref{eq:lb-reconst-dist})
and the results in Theorems \ref{thm:DRF-Gaussian-Matrix}-\ref{thm:DRF-vector},
we summarize the asymptotic bounds on the reconstruction distortion
as follows. Under Assumptions I, the reconstruction distortion of
scalar quantization is bounded by \begin{align*}
c_{lb}^{2}\frac{\pi\sqrt{3}}{2} & \le\underset{R\rightarrow\infty}{\lim}\underset{\left(K,m,N\right)\rightarrow\infty}{\lim}\frac{2^{2R}}{K}\mathrm{E}\left[\left\Vert \hat{\mathbf{X}}-\mathbf{X}\right\Vert _{2}^{2}\right]\\
 & \le\begin{cases}
c_{sp}^{2}\frac{\pi\sqrt{3}}{2} & \mathrm{for\; subspace\; algorithm}\\
c_{bp}^{2}\frac{\pi\sqrt{3}}{2} & \mathrm{for\; basis\; pursuit\; algorithm}\end{cases},\end{align*}
and the reconstruction distortion of uniform scalar quantization is
bounded by \begin{align*}
c_{lb}^{2}\frac{4\log2}{3} & \le\underset{R\rightarrow\infty}{\lim}\underset{\left(K,m,N\right)\rightarrow\infty}{\lim}\frac{2^{2R}}{KR}\mathrm{E}\left[\left\Vert \hat{\mathbf{X}}-\mathbf{X}\right\Vert _{2}^{2}\right]\\
 & \le\begin{cases}
c_{sp}^{2}\frac{4\log2}{3} & \mathrm{for\; subspace\; algorithm}\\
c_{bp}^{2}\frac{4\log2}{3} & \mathrm{for\; basis\; pursuit\; algorithm}\end{cases}.\end{align*}
Suppose that Assumption II holds. The reconstruction distortions for
scalar quantization and uniform scalar quantization are respectively
bounded by\begin{align*}
c_{lb}^{2}\frac{\pi\sqrt{3}}{2}\mu_{1} & \le\underset{R\rightarrow\infty}{\lim\inf}\frac{2^{2R}}{K}\mathrm{E}\left[\left\Vert \hat{\mathbf{X}}-\mathbf{X}\right\Vert _{2}^{2}\right]\\
 & \le\underset{R\rightarrow\infty}{\lim\sup}\frac{2^{2R}}{K}\mathrm{E}\left[\left\Vert \hat{\mathbf{X}}-\mathbf{X}\right\Vert _{2}^{2}\right]\\
 & \le\begin{cases}
c_{sp}^{2}\frac{\pi\sqrt{3}}{2}\mu_{2} & \mathrm{for\; subspace\; algorithm}\\
c_{bp}^{2}\frac{\pi\sqrt{3}}{2}\mu_{2} & \mathrm{for\; basis\; pursuit\; algorithm}\end{cases}\end{align*}
and \[
c_{lb}^{2}\frac{4\log2}{3}\mu_{1}\le\underset{R\rightarrow\infty}{\lim\inf}\frac{2^{2R}}{KR}\mathrm{E}\left[\left\Vert \hat{\mathbf{X}}-\mathbf{X}\right\Vert _{2}^{2}\right].\]
Given the encoding rate $R$ per measurement, the reconstruction distortion
of the optimal scalar quantizer is bounded as \begin{align*}
c_{lb}^{2}\frac{\pi e}{6} & \le\underset{R\rightarrow\infty}{\lim\inf}\underset{\left(K,m,N\right)\rightarrow\infty}{\lim\inf}\frac{2^{2R}}{K}\mathrm{E}\left[\left\Vert \hat{\mathbf{X}}-\mathbf{X}\right\Vert _{2}^{2}\right]\\
 & \le\underset{R\rightarrow\infty}{\lim\sup}\underset{\left(K,m,N\right)\rightarrow\infty}{\lim\sup}\frac{2^{2R}}{K}\mathrm{E}\left[\left\Vert \hat{\mathbf{X}}-\mathbf{X}\right\Vert _{2}^{2}\right]\\
 & \le\begin{cases}
c_{sp}^{2}\frac{\pi e}{3} & \mathrm{for\; subspace\; algorithm}\\
c_{bp}^{2}\frac{\pi e}{3} & \mathrm{for\; basis\; pursuit\; algorithm}\end{cases}.\end{align*}
The bounds for reconstruction distortion associated with vector quantization
are given by \begin{align*}
 & c_{lb}^{2}\left(1-\delta_{K}\right)\left(1+o_{K}\left(1\right)\right)\\
 & \le\underset{R\rightarrow\infty}{\lim\inf}\frac{2^{2Rm/K}}{K}\mathrm{E}\left[\left\Vert \hat{\mathbf{X}}-\mathbf{X}\right\Vert _{2}^{2}\right]\\
 & \le\underset{R\rightarrow\infty}{\lim\sup}\frac{2^{2R}}{m}\mathrm{E}\left[\left\Vert \hat{\mathbf{X}}-\mathbf{X}\right\Vert _{2}^{2}\right]\\
 & \le\begin{cases}
c_{sp}^{2}\left(1+\delta_{K}\right)\left(1+o_{m}\left(1\right)\right) & \mathrm{for\; subspace\; algorithm}\\
c_{bp}^{2}\left(1+\delta_{K}\right)\left(1+o_{m}\left(1\right)\right) & \mathrm{for\; basis\; pursuit\; algorithm}\end{cases},\end{align*}
and \begin{align*}
 & \underset{R\rightarrow\infty}{\lim\sup}\frac{2^{2R}}{K}\mathrm{E}\left[\left\Vert \hat{\mathbf{X}}-\mathbf{X}\right\Vert _{2}^{2}\right]\\
 & \le\begin{cases}
c_{sp}^{2}\frac{\pi\sqrt{3}}{2}\mu_{2} & \mathrm{for\; subspace\; algorithm}\\
c_{bp}^{2}\frac{\pi\sqrt{3}}{2}\mu_{2} & \mathrm{for\; basis\; pursuit\; algorithm}\end{cases}.\end{align*}

It is worth noting that the upper bound (\ref{eq:ub-reconst-dist})
on reconstruction distortion may not be tight. Empirical experiments
show that this upper bound often significantly over-estimates the
reconstruction distortion \cite{Candes_Tao_ApplMath2006_Stable_Signal_Recovery,Dai_2008_Subspace_Pursuit}.

\section{\label{sec:Modified-Algorithm}Reconstruction Algorithms for Quantized
CS}

We present next modifications of BP and SP algorithms that take into
account quantization effects.

To describe these algorithms, we find the following notation useful.
Let $\hat{\mathbf{Y}}$ be the quantized measurement vector. Given
a vector $\hat{\mathbf{Y}}$, the corresponding quantization region
can be easily identified: the quantization region of vector quantization
$\mathcal{R}_{\hat{\mathbf{Y}}}$ is defined in (\ref{eq:def-quant-region});
that of scalar quantization is given by the Cartesian product of the
quantization regions for each coordinate, i.e., $\mathcal{R}_{\hat{\mathbf{Y}}}=\prod_{i=1}^{m}\mathcal{R}_{\hat{Y}_{i}}$
where $\mathcal{R}_{\hat{Y}_{i}}$ is the quantization region of $\hat{Y}_{i}$. 

Similar to the standard BP method, the reconstruction problem can
be now casted as \begin{equation}
\min\left\Vert \mathbf{x}\right\Vert _{1}\;\mathrm{subject\; to}\;\mathbf{\Phi}\mathbf{x}\in\mathcal{R}_{\hat{\mathbf{Y}}}.\label{eq:BP-modified}\end{equation}
 It can be verified that $\mathcal{R}_{\hat{\mathbf{Y}}}$ is a closed
convex set and therefore (\ref{eq:BP-modified}) is a convex optimization
problem and can be efficiently solved by linear programming techniques.

In order to adapt the SP algorithm to the quantization scenario at
hand, we describe first a geometric interpretation of the projection
operation in the SP algorithm. Given $\mathbf{y}\in\mathbb{R}^{m}$
and $\mathbf{\Phi}_{T}\in\mathbb{R}^{m\times\left|T\right|}$, suppose
that $\mathbf{\Phi}_{T}$ has full column rank, in other words, suppose
that the columns of $\mathbf{\Phi}_{T}$ are linearly independent.
The projection operation in (\ref{eq:def-proj}) is equivalent to
the optimization problem \begin{equation}
\underset{\mathbf{x}\in\mathbb{R}^{\left|T\right|}}{\min}\left\Vert \mathbf{y}-\mathbf{\Phi}_{T}\mathbf{x}\right\Vert _{2}^{2}.\label{eq:proj-optim}\end{equation}
 Let $\mathbf{x}^{*}$ be the solution of the quadratic optimization
problem (\ref{eq:proj-optim}). Then functions (\ref{eq:def-proj}-\ref{eq:def-proj-coeff})
are equivalent to $\mathrm{proj}\left(\mathbf{y},\mathbf{\Phi}_{T}\right)=\mathbf{\Phi}_{T}\mathbf{x}^{*}$,
$\mathrm{resid}\left(\mathbf{y},\mathbf{\Phi}_{T}\right)=\mathbf{y}-\mathbf{\Phi}_{T}\mathbf{x}^{*}$
and $\mathrm{pcoeff}\left(\mathbf{y},\mathbf{\Phi}_{T}\right)=\mathbf{x}^{*}$.

The modified SP algorithm is based on the above geometric interpretation.
More precisely, we use the following definition. 
\begin{definitn}
\label{def:x-y-stars}For given $\mathbf{\Phi}_{T}\in\mathbb{R}^{m\times\left|T\right|}$,
$\hat{\mathbf{Y}}$ and $\mathcal{R}_{\hat{\mathbf{Y}}}$, define\begin{align}
\mathcal{Q} & :=\left\{ \left(\mathbf{x},\mathbf{y}\right)\in\mathbb{R}^{\left|T\right|}\times\mathcal{R}_{\hat{\mathbf{Y}}}:\right.\nonumber \\
 & \left.\left\Vert \mathbf{y}-\mathbf{\Phi}_{T}\mathbf{x}\right\Vert _{2}\le\left\Vert \mathbf{y}^{\prime}-\mathbf{\Phi}_{T}\mathbf{x}^{\prime}\right\Vert _{2}\;\forall\left(\mathbf{x}^{\prime},\mathbf{y}^{\prime}\right)\in\mathbb{R}^{\left|T\right|}\times\mathcal{R}_{\hat{\mathbf{Y}}}\right\} ,\label{eq:def-proj-q-plane}\end{align}
 and \begin{equation}
\left(\tilde{\mathbf{x}},\tilde{\mathbf{y}}\right)=\underset{\left(\mathbf{x},\mathbf{y}\right)\in\mathcal{Q}}{\arg\min}\left\Vert \mathbf{y}-\hat{\mathbf{Y}}\right\Vert _{2}.\label{eq:def-proj-q-solution}\end{equation}

\end{definitn}
\vspace{0.01in}

It can be verified that the pair $\left(\tilde{\mathbf{x}},\tilde{\mathbf{y}}\right)$
is well defined. See Appendix \ref{sub:pf-well-define} for details.

This definition is introduced to identify the best approximation for
$\hat{\mathbf{Y}}$ among multiple points in $\mathcal{R}_{\hat{\mathbf{Y}}}$
that minimize $\left\Vert \mathbf{y}-\mathbf{\Phi}_{T}\mathbf{x}\right\Vert _{2}$.
Based on this definition, we replace the $\mathrm{resid}$ and $\mathrm{pcoeff}$
functions in Algorithm \ref{alg:Subspace-Pursuit-Algorithm} with
new functions \[
\mathrm{resid}^{\left(q\right)}\left(\hat{\mathbf{Y}},\mathbf{\Phi}_{T}\right):=\tilde{\mathbf{y}}-\mathbf{\Phi}_{T}\tilde{\mathbf{x}}\]
 and \[
\mathrm{pcoeff}^{\left(q\right)}\left(\hat{\mathbf{Y}},\mathbf{\Phi}_{T}\right):=\tilde{\mathbf{x}},\]
 where the superscript $\left(q\right)$ emphasizes that these definitions
are for the quantized case. This gives the modified SP algorithm.

The advantage of the modified algorithms are verified by the simulation
results presented in the next section.

\section{\label{sec:Empirical-Results}Empirical Results}

We performed extensive computer simulations in order to compare the
performance of different quantizers and different reconstruction algorithms
empirically. The parameters used in our simulations are $m=128$,
$N=256$ and $K=6$. Given these parameters, we generated realizations
of $m\times N$ sampling matrices from the i.i.d. standard Gaussian
ensemble and normalize the columns to have unit $l_{2}$-norm. We
also selected a support set $T$ of size $\left|T\right|=K$ uniformly
at random, generated the entries supported by $T$ from the standard
i.i.d. Gaussian distribution and set all other entries to zero. We
let the quantization rates vary from two to six bits. For each quantization
rate, we used Lloyd's algorithm (Section \ref{sub:Scalar-Quantization})
to obtain a nonuniform quantizer and employed brute-force search to
find the optimal uniform quantizer. To test different quantizers and
reconstruction algorithms, we randomly generated $\mathbf{\Phi}$
and $\mathbf{x}$ independently a thousand times. For each realization,
we calculated the measurements $\mathbf{Y}$, the quantized measurements
$\hat{\mathbf{Y}}$ and the reconstructed signal $\hat{\mathbf{X}}$.

Fig. \ref{fig:dist-measurements} compares uniform and uniform quantizers
with respect to measurement distortion. Though the quantization rates
in our experiments are relatively small, the simulation results are
consistent with the asymptotic results in Theorem \ref{thm:DRF-Gaussian-Matrix}:
nonuniform quantization is better than uniform quantization and the
gain increases with the quantization rate. Fig. \ref{fig:Rec-Dist-Standard}
compares the reconstruction distortion of the standard BP and SP algorithms.
The comparison of the modified algorithms is given in Fig. \ref{fig:dist-reconstruction}.
The modified algorithms reduce the reconstruction distortion significantly.
When the quantization rate is six bits, the reconstruction distortion
of the modified algorithms is roughly one tenth of that of the standard
algorithms. Furthermore, for both the standard and modified algorithms,
the reconstruction distortion given by SP algorithms is much smaller
than that of BP methods. Note that the computational complexity of
the SP algorithms is also smaller than that of the BP methods, which
shows clear advantages for using SP algorithms in conjunction with
quantized CS data. An interesting phenomenon occurs for the case of
the modified BP method: although nonuniform quantization gives smaller
measurement distortion, the corresponding reconstruction distortion
is actually slightly larger than that of uniform quantization. We
do not have solid analytical arguments to completely explain this
somewhat counter-intuitive fact.

\appendix

\subsection{\label{sub:Pf-Thm-Gaussian-scalar}Proof of Theorem \ref{thm:DRF-Gaussian-Matrix}}

Let $T=\left\{ 1\le j\le N:\; X_{j}\ne0\right\} $ be the support
set of $\mathbf{x}$, i.e., $x_{i}\ne0$ for all $i\in T$ and $x_{j}=0$
for all $j\notin T$. It is easy to show that for all $1\le i\le m$
and $T\subset\left\{ 1,\cdots,N\right\} $ such that $\left|T\right|=K$,
\[
\mathrm{E}\left[\sum_{j\in T}A_{i,j}X_{j}\right]=0\]
 and \[
\mathrm{E}\left[\left(\sum_{j\in T}A_{i,j}X_{j}\right)^{2}\right]=K.\]

According to the Central Limit Theorem, the distribution of $\frac{1}{\sqrt{K}}\sum_{j\in T}A_{i,j}X_{j}$
converges weakly to the standard Gaussian distribution as $K\rightarrow\infty$.
This can be verified by the facts that $A_{i,j}X_{j}$s are independent
and identically distributed, and that the moment generating function
of $A_{i,j}X_{j}$ is well defined. As a result, the distribution
of $\sqrt{\frac{m}{K}}Y_{i}$ converges weakly to the standard Gaussian
distribution as $K,m,N\rightarrow\infty$.

We apply a scalar quantizer with $2^{R}$ levels to the random variable
$\sqrt{\frac{m}{K}}Y_{i}$. In this case, one has \begin{align}
 & \frac{1}{K}\mathrm{E}\left[\left\Vert \hat{\mathbf{Y}}-\mathbf{Y}\right\Vert _{2}^{2}\right]\nonumber \\
 & =\frac{1}{m}\frac{m}{K}\mathrm{E}\left[\sum_{i=1}^{m}\left(\hat{Y}_{i}-Y_{i}\right)^{2}\right]\nonumber \\
 & =\frac{1}{m}\sum_{i=1}^{m}\mathrm{E}\left[\left(\sqrt{\frac{m}{K}}\hat{Y}_{i}-\sqrt{\frac{m}{K}}Y_{i}\right)^{2}\right]\nonumber \\
 & =\mathrm{E}\left[\left(\sqrt{\frac{m}{K}}\hat{Y}_{i}-\sqrt{\frac{m}{K}}Y_{i}\right)^{2}\right],\label{eq:T1-02}\end{align}
 where the last line represents the distortion of quantizing $\sqrt{\frac{m}{K}}Y_{i}$.
Note that the distortion-rate function for scalar quantization of
a Gaussian random variable is given by \begin{equation}
\underset{R\rightarrow\infty}{\lim}2^{2R}D_{g}^{*}\left(R\right)=\frac{\pi\sqrt{3}}{2}\sigma^{2},\label{eq:DRF-scalar-Gaussian}\end{equation}
 where $\sigma^{2}$ is the variance of the underlying Gaussian source
(see \cite{Zadar1964_thesis_nonuniform_quantization} for a detailed
proof of this result). We then have \[
\underset{R\rightarrow\infty}{\lim}\;\underset{\left(K,m,N\right)\rightarrow\infty}{\lim}\frac{2^{2R}}{K}D^{*}\left(R\right)=\underset{R\rightarrow\infty}{\lim}2^{2R}D_{g}^{*}\left(R\right)=\frac{\pi\sqrt{3}}{2},\]
 which completes the proof of (\ref{eq:DRF-Gaussian-non-uniform}).

Consider a uniform quantizer with codebook $\mathcal{C}_{u}$, such
that $\left|\mathcal{C}_{u}\right|=2^{R}$, and apply the corresponding
uniform quantizer to the random variable $\sqrt{\frac{m}{K}}Y_{i}$.
It was shown in \cite{Neuhoff2001_IT_optimal_uniform_scalar_quantization}
that the distortion-rate function of uniform scalar quantizers of
a Gaussian random variable equals \begin{equation}
\underset{R\rightarrow\infty}{\lim}\frac{2^{2R}}{R}D_{u,g}^{*}\left(R\right)=\frac{4}{3}\sigma^{2}\log2.\label{eq:DRF-uniform-scalar-Gaussian}\end{equation}
 It is clear that \[
\underset{R\rightarrow\infty}{\lim}\;\underset{\left(K,m,N\right)\rightarrow\infty}{\lim}\frac{2^{2R}}{KR}D_{u}^{*}\left(R\right)=\underset{R\rightarrow\infty}{\lim}\frac{2^{2R}}{R}D_{u,g}^{*}\left(R\right)=\frac{4}{3}\log2,\]
 This proves Theorem \ref{thm:DRF-Gaussian-Matrix}.

\subsection{\label{sub:Pf-lb-ub}Proof of Theorems \ref{thm:DRF-Matrix} and
\ref{thm:DRF-vector}}

For completeness, let us first briefly review the key results used
for deriving the asymptotic distortion-rate function for CS vector
quantization. Suppose the source $\mathbf{Y}\in\mathbb{R}^{k}$ has
probability density function $f\left(\mathbf{y}\right)$. Let $\mathcal{R}\subset\mathbb{R}^{k}$
be a quantization region and $\bm{\omega}\in\mathcal{C}$ be the corresponding
quantization level. The corresponding normalized moment of inertia
(NMI) is defined as \[
m\left(\mathcal{R}\right)=\frac{\frac{1}{k}\int_{\mathcal{R}}\left\Vert \mathbf{y}-\bm{\omega}\right\Vert _{2}^{2}f\left(\mathbf{y}\right)d\mathbf{y}}{\left(\int_{\mathcal{R}}d\mathbf{y}\right)^{1+2/k}}.\]
The optimal NMI equals \[
m_{k}^{*}=\underset{\mathcal{R}\subset\mathbb{R}^{k}}{\inf}m\left(\mathcal{R}\right),\]
 only depends on the number of dimensions: $m_{k}^{*}=c_{k}$ with
$c_{k}=\frac{1}{12}$ when $k=1$ and $c_{k}\rightarrow\frac{1}{2\pi e}$
when $k\rightarrow\infty$. Thus the distortion rate function satisfies
\begin{align}
\underset{R\rightarrow\infty}{\lim}\frac{2^{R}}{k}D\left(R\right) & =\int\frac{f\left(\mathbf{y}\right)}{\lambda_{k}^{2/k}\left(\mathbf{y}\right)}m_{k}^{*}d\mathbf{y},\label{eq:distortion-point-density-fn}\end{align}
 where $R$ is the quantization rate per dimension, and $\lambda_{k}\left(\mathbf{y}\right)$
denotes the point density function. In this case, the integral \[
\int_{\mathcal{M}}\lambda_{k}\left(\mathbf{y}\right)d\mathbf{y}\]
 gives the fraction of quantization levels belonging to $\mathcal{M}$
for all measurable sets $\mathcal{M}\subset\mathbb{R}^{k}$. For simplicity,
we have assumed that $\lambda_{k}\left(\mathbf{y}\right)$ is continuous
on $\mathbb{R}^{k}$. For fixed $m_{k}^{*}$, the problem of designing
an asymptotically optimal quantizer can be reduced to the problem
of finding the point density function $\lambda_{k}^{*}\left(\mathbf{y}\right)$
that minimizes (\ref{eq:distortion-point-density-fn}). By Hölder's
inequality, the optimal point density function is given by \[
\lambda_{k}^{*}\left(\mathbf{y}\right)=\frac{f^{k/\left(k+2\right)}\left(\mathbf{y}\right)}{\int f^{k/\left(k+2\right)}\left(\mathbf{y}\right)\cdot d\mathbf{y}},\]
 and the asymptotic distortion rate function is therefore \begin{align}
\underset{R\rightarrow\infty}{\lim}\frac{2^{R}}{k}D^{*}\left(R\right) & =c_{k}\left(\int f^{k/\left(k+2\right)}\left(\mathbf{y}\right)\cdot d\mathbf{y}\right)^{\frac{k+2}{k}}.\label{eq:DRF-in-point-density-fn}\end{align}
 If the source $\mathbf{Y}$ is Gaussian distributed with covariance
matrix $\mathbf{\Sigma}>0$, then the asymptotic distortion rate function
(\ref{eq:DRF-in-point-density-fn}) can be explicitly evaluated as
\begin{align}
\underset{R\rightarrow\infty}{\lim}\frac{2^{R}}{k}D^{*}\left(R\right) & =c_{k}\left|2\pi\mathbf{\Sigma}\right|^{\frac{1}{k}}\left(\frac{k+2}{k}\right)^{\frac{k+2}{2}}\label{eq:DRF-point-density-Gaussian}\\
 & =\left|\mathbf{\Sigma}\right|^{\frac{1}{k}}\left(1+o_{K}\left(1\right)\right),\nonumber \end{align}
 where $o_{K}\left(1\right)\rightarrow0$ as $K\rightarrow\infty$,
and the last equality follows from the fact that $c_{k}\rightarrow\frac{1}{2\pi e}$
and $\left(\frac{k+2}{2}\right)^{\frac{k+2}{2}}\rightarrow e$ as
$k\rightarrow\infty$.

We present next the key results used for proving the upper bounds
in (\ref{eq:DRF-non-uniform}) and (\ref{eq:ub-DRF-VQ-1}). 
\begin{prop}
\label{pro:ub-mismatch} Let $\mathbf{Y}_{0}\in\mathbb{R}^{k}$ be
a Gaussian random vector with zero mean and covariance matrix $\mathbf{\Sigma}_{0}$.
Let $\left\{ \mathfrak{q}_{R}\left(\cdot\right)\right\} $, where
the subscript $R$ denotes the quantization rate, be a sequence of
quantizers designed to achieve the asymptotic distortion rate function
for Gaussian source $\mathcal{N}\left(\mathbf{0},\mathbf{\Sigma}_{1}\right)$
with $\mathbf{0}<\mathbf{\Sigma}_{1}\in\mathbb{R}^{k\times k}$. Apply
$\mathfrak{q}_{R}\left(\cdot\right)$ to $\mathbf{Y}_{0}$. If $\mathbf{\Sigma}_{0}<\mathbf{\Sigma}_{1}$,
then\begin{align}
 & \underset{R\rightarrow\infty}{\lim}\frac{2^{2R}}{k}\mathrm{E}_{Y_{0}}\left[\left\Vert \mathbf{Y}_{0}-\mathfrak{q}_{R}\left(\mathbf{Y}_{0}\right)\right\Vert _{2}^{2}\right]\nonumber \\
 & \le c_{k}\left(2\pi\mathbf{\Sigma}_{1}\right)^{\frac{1}{k}}\left(\frac{k+2}{k}\right)^{\frac{k+2}{2}}.\label{eq:ub-03}\end{align}
 \end{prop}
\begin{proof}
First assume that $\mathbf{0}<\mathbf{\Sigma}_{0}$. Let $f_{0}\left(\mathbf{y}\right)$
and $f_{1}\left(\mathbf{y}\right)$ be the probability density functions
for $\mathbf{Y}_{0}$ and $\mathbf{Y}_{1}$, respectively. Denote
$\mathrm{E}_{Y_{0}}\left[\left\Vert \mathbf{Y}_{0}-\mathfrak{q}_{R}\left(\mathbf{Y}_{0}\right)\right\Vert _{2}^{2}\right]$
by $D\left(R\right)$. It is clear that \begin{align}
 & \underset{R\rightarrow\infty}{\lim}\frac{2^{R}}{k}D\left(R\right)\nonumber \\
 & =c_{k}\int\frac{f_{0}\left(\mathbf{y}\right)}{\left(\lambda_{k,1}^{*}\left(\mathbf{y}\right)\right)^{2/k}}d\mathbf{y}\nonumber \\
 & =c_{k}\int\frac{f_{0}\left(\mathbf{y}\right)}{f_{1}^{2/\left(k+2\right)}\left(\mathbf{y}\right)}d\mathbf{y}\cdot\left(\int f_{1}^{k/\left(k+2\right)}\left(\mathbf{y}\right)d\mathbf{y}\right)^{\frac{2}{k}}.\label{eq:ub-01}\end{align}
 We upper bound the first integral as follows \begin{align}
 & \int\frac{f_{0}\left(\mathbf{y}\right)}{f_{1}^{2/\left(k+2\right)}\left(\mathbf{y}\right)}d\mathbf{y}\nonumber \\
 & =\frac{\left|2\pi\mathbf{\Sigma}_{1}\right|^{\frac{1}{k+2}}}{\left|2\pi\mathbf{\Sigma}_{0}\right|^{1/2}}\int\exp\left\{ -\frac{1}{2}\mathbf{y}^{*}\left(\mathbf{\Sigma}_{0}^{-1}-\frac{2}{k+2}\mathbf{\Sigma}_{1}^{-1}\right)\mathbf{y}\right\} d\mathbf{y}\nonumber \\
 & \overset{\left(a\right)}{=}\frac{\left|2\pi\mathbf{\Sigma}_{1}\right|^{\frac{1}{k+2}}}{\left|2\pi\mathbf{\Sigma}_{0}\right|^{1/2}}\frac{\left|2\pi\bm{\Sigma}_{0}\right|^{1/2}}{\left|\mathbf{I}_{k}-\frac{2}{k+2}\bm{\Sigma}_{0}\bm{\Sigma}_{1}^{-1}\right|^{1/2}}\nonumber \\
 & \overset{\left(b\right)}{\le}\left|2\pi\bm{\Sigma}_{1}\right|^{\frac{1}{k+2}}\left(\frac{k+2}{k}\right)^{\frac{k}{2}}\nonumber \\
 & =\int f_{1}^{\frac{k}{k+2}}\left(\bm{x}\right)d\bm{x},\label{eq:ub-02}\end{align}
 where $\left(a\right)$ holds because \begin{align*}
 & \mathbf{\Sigma}_{0}^{-1}-\frac{2}{k+2}\bm{\Sigma}_{1}^{-1}\\
 & =\bm{\Sigma}_{0}^{-1}\left(\mathbf{I}_{k}-\frac{2}{k+2}\bm{\Sigma}_{0}\bm{\Sigma}_{1}^{-1}\right)\\
 & =\left[\left(\mathbf{I}_{k}-\frac{2}{k+2}\bm{\Sigma}_{0}\bm{\Sigma}_{1}^{-1}\right)^{-1}\bm{\Sigma}_{0}\right]^{-1},\end{align*}
 and $\left(b\right)$ follows from the assumption $\mathbf{\Sigma}_{0}<\mathbf{\Sigma}_{1}$.
Substituting (\ref{eq:ub-02}) into (\ref{eq:ub-01}), one obtains\begin{align*}
 & \underset{R\rightarrow\infty}{\lim}\frac{2^{R}}{k}D\left(R\right)\\
 & \le c_{k}\left(\int f_{1}^{k/\left(k+2\right)}\left(\mathbf{y}\right)d\mathbf{y}\right)^{\frac{k+2}{k}}\\
 & =c_{k}\left|2\pi\mathbf{\Sigma}_{1}\right|^{\frac{1}{k}}\left(\frac{k+2}{k}\right)^{\frac{k+2}{2}},\end{align*}
 which will be used to prove the upper bounds in (\ref{eq:DRF-non-uniform})
and (\ref{eq:ub-DRF-VQ-1}).

Suppose that $\left|\mathbf{\Sigma}_{0}\right|=0$ (some of the eigenvalues
of $\mathbf{\Sigma}_{0}$ are zero). Since $\mathbf{\Sigma}_{0}<\mathbf{\Sigma}_{1}$,
when $\epsilon>0$ is sufficiently small, we have $\mathbf{0}<\mathbf{\Sigma}_{\epsilon}:=\mathbf{\Sigma}_{0}+\epsilon\mathbf{I}<\mathbf{\Sigma}_{1}$.
Let $f_{\epsilon}\left(\mathbf{y}\right)$ be the probability density
function of Gaussian vector with zero mean and variance $\mathbf{\Sigma}_{\epsilon}$.
Then, \begin{align*}
 & \underset{R\rightarrow\infty}{\lim}\frac{2^{R}}{k}D\left(R\right)\\
 & =c_{k}\int\frac{f_{0}\left(\mathbf{y}\right)}{\left(\lambda_{k,1}^{*}\left(\mathbf{y}\right)\right)^{2/k}}d\mathbf{y}\\
 & =c_{k}\int\frac{\underset{\epsilon\rightarrow0}{\lim}f_{\epsilon}\left(\mathbf{y}\right)}{\left(\lambda_{k,1}^{*}\left(\mathbf{y}\right)\right)^{2/k}}d\mathbf{y}\\
 & \overset{\left(c\right)}{\le}c_{k}\underset{\epsilon\rightarrow0}{\lim\inf}\int\frac{f_{\epsilon}\left(\mathbf{y}\right)}{\left(\lambda_{k,1}^{*}\left(\mathbf{y}\right)\right)^{2/k}}d\mathbf{y}\\
 & \overset{\left(d\right)}{\le}c_{k}\left|2\pi\mathbf{\Sigma}_{1}\right|^{\frac{1}{k}}\left(\frac{k+2}{k}\right)^{\frac{k+2}{2}},\end{align*}
 where $\left(c\right)$ follows from Fatou's lemma \cite{Royden1988_real_analysis},
and $\left(d\right)$ follows from the first part of this proof. This
proves the proposition. 
\end{proof}

\subsubsection{Lower Bounds for Scalar Quantization}

$ $

We prove the lower bound in (\ref{eq:DRF-non-uniform}). Given Assumptions
II, each $Y_{i}$, $1\le i\le m$, is a linear combination of Gaussian
random variables, and therefore each $Y_{i}$ is a Gaussian random
variable itself. For a given $i$ and a given $T$, the mean and the
variance of $Y_{i}$ are $\mathrm{E}\left[Y_{i}\right]=0$ and $\sigma_{i,T}^{2}=\mathrm{E}\left[Y_{i}^{2}\right]=\sum_{j\in T}\varphi_{i,j}^{2}$,
respectively. The variance depends on the row index $i$ and the support
set $T$. We calculate the average variance across all rows and all
support sets as \begin{align}
\bar{\sigma}^{2} & =\frac{1}{m}\sum_{i=1}^{m}\left(\frac{1}{{N \choose K}}\sum_{T}\;\sum_{j\in T}\varphi_{i,j}^{2}\right)\nonumber \\
 & =\frac{1}{m}\frac{1}{{N \choose K}}\sum_{T}\;\sum_{j\in T}\left(\sum_{i=1}^{m}\varphi_{i,j}^{2}\right)\nonumber \\
 & \overset{\left(a\right)}{=}\frac{1}{m}\frac{1}{{N \choose K}}\sum_{j=1}^{N}\left(\sum_{T:\; j\in T}\left\Vert \bm{\varphi}_{j}\right\Vert _{2}^{2}\right)\nonumber \\
 & \overset{\left(b\right)}{=}\frac{1}{m}\frac{1}{{N \choose K}}\sum_{j=1}^{N}{N-1 \choose K-1}\left\Vert \bm{\varphi}_{j}\right\Vert _{2}^{2}\nonumber \\
 & \overset{\left(c\right)}{=}\frac{K}{m}\frac{1}{N}\sum_{j=1}^{N}\left\Vert \bm{\varphi}_{j}\right\Vert _{2}^{2}\nonumber \\
 & \overset{\left(d\right)}{=}\frac{K}{m}\mu_{1},\label{eq:sigma-bar}\end{align}
 where 
\begin{description}
\item [{{$\left(a\right)$}}] is obtained by exchanging the sums over
$T$ and $j$, 
\item [{{$\left(b\right)$}}] holds because for any given $1\le j\le N$,
there are ${N-1 \choose K-1}$ many subsets $T$ containing the index
$j$, 
\item [{{$\left(c\right)$}}] is due to the fact that ${N-1 \choose K-1}/{N \choose K}=K/N$, 
\item [{{$\left(d\right)$}}] follows from the definition (\ref{eq:def-mu-1}). 
\end{description}
\vspace{0.01in}

Suppose that one deals with the ideal case: the support set $T$ is
known before taking the measurements; and for different values of
$i$ and $T$, we are allowed to use different quantizers. Given $i$
and $T$, we apply the optimal quantizer for the Gaussian random variable
$\sqrt{\frac{m}{K}}Y_{i}$, so that the quantization distortion of
$Y_{i}$ satisfies \[
\underset{R\rightarrow\infty}{\lim}2^{2R}D_{i,T}^{*}\left(R\right)=\frac{\pi\left(\frac{m}{K}\sigma_{i,T}^{2}\right)}{2}\sqrt{3},\]
 which is a direct application of (\ref{eq:DRF-in-point-density-fn})
with $k=1$. Taking the average over all $i$ and all $T$ gives \begin{align*}
 & \underset{R\rightarrow\infty}{\lim}\frac{1}{m}\sum_{i=1}^{m}\mathrm{E}_{T}\left[2^{2R}D_{i,T}^{*}\left(R\right)\right]\\
 & =\frac{1}{m}\sum_{i=1}^{m}\frac{1}{{T \choose K}}\sum_{T}\left(\underset{R\rightarrow\infty}{\lim}2^{2R}D_{i,T}^{*}\left(R\right)\right)\\
 & =\frac{1}{m}\frac{1}{{T \choose K}}\sum_{i=1}^{m}\sum_{T}\left(\frac{\pi\left(\frac{m}{K}\sigma_{i,T}^{2}\right)}{2}\sqrt{3}\right)\\
 & =\frac{\pi\mu_{1}}{2}\sqrt{3},\end{align*}
 where the last equality follows from (\ref{eq:sigma-bar}).

However, the support set $T$ is unknown before taking the measurements.
Furthermore, the same quantizer has to be employed for different choices
of $i$ and $T$. Thus, for every $R$, $i$ and $T$, $\mathrm{E}_{Y_{i}}\left[\frac{m}{K}\left|Y_{i}-\hat{Y}_{i}\right|^{2}\right]\ge D_{i,T}^{*}\left(R\right)$.
As a result,\begin{align*}
 & \underset{R\rightarrow\infty}{\lim\inf}\frac{2^{2R}}{K}\mathrm{E}_{T}\left[\mathrm{E}_{Y}\left[\left\Vert \hat{\mathbf{Y}}-\mathbf{Y}\right\Vert _{2}^{2}\right]\right]\\
 & =\underset{R\rightarrow\infty}{\lim\inf}\frac{2^{2R}}{{N \choose T}}\sum_{T}\frac{1}{m}\sum_{i=1}^{m}\frac{m}{K}\mathrm{E}_{Y}\left[\left(\hat{y}_{i}-y_{i}\right)^{2}\right]\\
 & \ge\underset{R\rightarrow\infty}{\lim\inf}\frac{2^{2R}}{{N \choose T}}\sum_{T}\frac{1}{m}\sum_{i=1}^{m}D_{i,T}^{*}\left(R\right)\\
 & =\frac{\pi\mu_{1}}{2}\sqrt{3}.\end{align*}
 Since the above derivation is valid for all $K$, $m$ and $N$,
the claim in (\ref{eq:DRF-non-uniform}) holds.

The result in (\ref{eq:DRF-uniform}) for uniform quantizers can be
proved using similar arguments. For the ideal case, given $i$ and
$T$, apply the optimal \emph{uniform} quantizer for the standard
Gaussian random variable to $\sqrt{\frac{m}{K}}y_{i}$. The corresponding
distortion rate function for this case was characterized in \cite{Neuhoff2001_IT_optimal_uniform_scalar_quantization}
and s given by \[
\underset{R\rightarrow\infty}{\lim}2^{2R}D_{u,i,T}^{*}\left(R\right)=\frac{4}{3}\sigma_{i,T}^{2}\ln2.\]
 Therefore, \begin{align*}
 & \underset{R\rightarrow\infty}{\lim\inf}\frac{2^{2R}}{K}\mathrm{E}_{T}\left[\mathrm{E}_{Y}\left[\left\Vert \hat{\mathbf{Y}}-\mathbf{Y}\right\Vert _{2}^{2}\right]\right]\\
 & \ge\frac{4}{3}\mu_{1}\ln2,\end{align*}
 which completes the proof of (\ref{eq:DRF-uniform}).

\subsubsection{The Upper Bound for Scalar Quantization}

$ $

By the definition of $\mu_{2}$ in (\ref{eq:def-mu-2}), the variance
of the Gaussian random variable $\sqrt{\frac{m}{K}}Y_{i}$ is upper
bounded by $\mu_{2}$ uniformly for all $i$ and all $T$. For each
quantization rate $R$, we design the optimal quantizer for a Gaussian
source with variance $\mu_{2}$ and apply this quantizer to quantize
all components of $\mathbf{Y}$. Using (\ref{eq:ub-03}), one can
show that the quantization distortion for all $i$ and $T$ satisfies\begin{align*}
 & \underset{R\rightarrow\infty}{\lim\sup}\frac{2^{2R}}{K}\mathrm{E}_{T}\mathrm{E}_{\mathbf{Y}}\left[\left\Vert \hat{\mathbf{Y}}-\mathbf{Y}\right\Vert _{2}^{2}\right]\\
 & \le\frac{\pi}{2}\mu_{2}\sqrt{3},\end{align*}
 which proves the upper bound in (\ref{eq:DRF-non-uniform}).

\subsubsection{The Lower Bound for Vector Quantization}

$ $

The basic idea for proving the lower bound in (\ref{eq:lb-DRF-VQ})
is similar to that behind (\ref{eq:DRF-non-uniform}). For each $T$,
a lower bound on the minimum achievable distortion is derived. The
average distortion taken over all the sets $T$ serves as a lower
bound of the overall distortion-rate function.

Suppose the ideal case where we have prior knowledge of $T\in{\left[N\right] \choose K}$.
We study the distortion rate function for every given $T$. The measurement
vector $\mathbf{Y}$ is Gaussian distributed with zero mean and covariance
matrix $\mathbf{\Phi}_{T}\mathbf{\Phi}_{T}^{*}$, where $\mathbf{\Phi}_{T}$
consists of the columns of $\mathbf{\Phi}$ indexed by $T$. The singular
value decomposition of $\mathbf{\Phi}_{T}\mathbf{\Phi}_{T}^{*}$ gives
$\mathbf{U}_{T}\mathbf{\Lambda}_{T}\mathbf{U}_{T}^{*}$, where $\mathbf{U}_{T}\in\mathbb{R}^{m\times m}$
has orthonormal columns and $\mathbf{\Lambda}_{T}=\mathrm{diag}\left(\lambda_{1},\lambda_{2},\cdots,\lambda_{m}\right)$
is the diagonal matrix formed by the singular values $\lambda_{1}\ge\lambda_{2}\ge\cdots\ge\lambda_{m}$.
Note that $\lambda_{i}\left(\mathbf{\Phi}_{T}^{*}\mathbf{\Phi}_{T}\right)=\lambda_{i}\left(\mathbf{\Phi}_{T}\mathbf{\Phi}_{T}^{*}\right)$
for $1\le i\le K$. According to Assumption III.1, the measurement
matrix $\mathbf{\Phi}$ satisfies the RIP with constant parameter
$\delta_{K}$ , which implies that $1-\delta_{K}\le\lambda_{i}\left(\mathbf{\Phi}_{T}^{*}\mathbf{\Phi}_{T}\right)\le1+\delta_{K}$
for all $1\le i\le K$. It can be concluded that $1-\delta_{K}\le\lambda_{i}\le1+\delta_{K}$
for $1\le i\le K$ and $\lambda_{i}=0$ for $K+1\le i\le m$. As a
result, $\mathbf{\Phi}_{T}\mathbf{\Phi}_{T}^{*}=\mathbf{U}_{T,K}\mathbf{\Lambda}_{T,K}\mathbf{U}_{T,K}^{*}$
where $\mathbf{U}_{T,K}\in\mathbb{R}^{m\times K}$ contains the first
$K$ columns of $\mathbf{U}_{T}$ and $\mathbf{\Lambda}_{T,K}\in\mathbb{R}^{K\times K}$
is the diagonal matrix formed by the $K$ largest singular values.
Denote the matrix formed by the last $m-K$ columns of $\mathbf{U}$
by $\mathbf{U}_{T,K}^{\perp}$: clearly, $\mathbf{U}_{T}=\left[\mathbf{U}_{T,K}|\mathbf{U}_{T,K}^{\perp}\right]$.

The best quantization strategy is to quantize $\bar{\mathbf{Y}}=\mathbf{U}_{T,K}^{*}\mathbf{Y}$
so that no quantization bit is used for the {}``trivial signal''
$\left(\mathbf{U}_{T,K}^{\perp}\right)^{*}\mathbf{Y}$. It is clear
that $\bar{\mathbf{Y}}\sim\mathcal{N}\left(\mathbf{0},\mathbf{\Lambda}_{T,K}\right)$
and $\mathbf{0}<\mathbf{\Lambda}_{T,K}$. The corresponding asymptotic
distortion rate function is therefore \begin{align*}
\underset{R\rightarrow\infty}{\lim}\frac{2^{2mR/K}}{K}D_{T}^{*}\left(R\right) & \overset{\left(\ref{eq:DRF-point-density-Gaussian}\right)}{=}c_{K}\left(2\pi\mathbf{\Lambda}_{T,K}\right)^{\frac{1}{K}}\left(\frac{K+2}{K}\right)^{\frac{K+2}{2}}\\
 & \ge\left(1-\delta_{K}\right)\left(1+o_{K}\left(1\right)\right),\end{align*}
 where the $2^{2mR/K}$ term comes from the fact that the total quantization
rate $mR$ is used to quantize a $K$-dimensional signal. Since this
lower bound is valid for all $T\in{\left[N\right] \choose K}$, we
have proved the lower bound in (\ref{eq:lb-DRF-VQ}).

\subsubsection{The Upper Bound for Vector Quantization}

$ $

Let $\epsilon>0$ be a small constant. Let $\left\{ \mathfrak{q}_{R}\left(\cdot\right)\right\} $
be a sequence of quantizers that approaches the asymptotic distortion
rate function for quantizing $\bar{\mathbf{Y}}\sim\mathcal{N}\left(\mathbf{0},\left(1+\delta_{K}+\epsilon\right)\mathbf{I}_{m}\right)$.
To prove the upper bound in (\ref{eq:ub-DRF-VQ-1}), apply the quantizer
sequence $\left\{ \mathfrak{q}_{R}\left(\cdot\right)\right\} $ to
$\mathbf{Y}$. For every $T\in{\left[N\right] \choose K}$, $\mathbf{Y}\sim\mathcal{N}\left(\mathbf{0},\mathbf{\Phi}_{T}\mathbf{\Phi}_{T}^{*}\right)$.
According to the Assumption III.1, $\mathbf{\Phi}_{T}\mathbf{\Phi}_{T}^{*}<\left(1+\delta_{K}+\epsilon\right)\mathbf{I}_{m}$.
Applying Proposition \ref{pro:ub-mismatch}, we have \begin{align*}
 & \underset{R\rightarrow\infty}{\lim}\frac{2^{2R}}{m}\mathrm{E}_{\mathbf{Y}}\left[\left\Vert \mathbf{Y}-\mathfrak{q}_{R}\left(\mathbf{Y}\right)\right\Vert _{2}^{2}\right]\\
 & \le\left(1+\delta_{K}+\epsilon\right)\left(1+o_{M}\left(1\right)\right).\end{align*}
 The upper bound in (\ref{eq:ub-DRF-VQ-1}) is proved by taking the
limit $\epsilon\downarrow0$.

\subsection{\label{sub:pf-well-define}The Existence and Uniqueness of $\left(\tilde{\mathbf{x}},\tilde{\mathbf{y}}\right)$
in Equation (\ref{eq:def-proj-q-solution})}

Consider the optimization problem \begin{equation}
\underset{\left(\mathbf{x},\mathbf{y}\right)\in\mathbb{R}^{\left|T\right|}\times\mathcal{R}_{\hat{\mathbf{Y}}}}{\min}\;\left\Vert \mathbf{y}-\mathbf{\Phi}_{T}\mathbf{x}\right\Vert _{2},\label{eq:well-def-01}\end{equation}
 which is equivalent to\begin{equation}
\underset{\left(\mathbf{x},\mathbf{y}\right)\in\mathbb{R}^{\left|T\right|}\times\mathcal{R}_{\hat{\mathbf{Y}}}}{\min}\;\left\Vert \left[-\mathbf{\Phi}_{T}\;\mathbf{I}\right]\left[\begin{array}{c}
\mathbf{x}\\
\mathbf{y}\end{array}\right]\right\Vert _{2}^{2}.\label{eq:well-def-02}\end{equation}
 Note that the objective function is convex and the constraint set
is convex and closed. The optimization problem (\ref{eq:well-def-02})
has at least one solution. Note that the matrix $\left[-\mathbf{\Phi}_{T}\;\mathbf{I}\right]$
does not have full row-rank. Hence, the solution may not be unique:
the set $\mathcal{Q}$ defined in (\ref{eq:def-proj-q-plane}) gives
all the possible solutions, and is convex and closed.

Let $\mathfrak{P}$ be the projection function from $\mathbb{R}^{\left|T\right|}\times\mathbb{R}^{m}$
to $\mathbb{R}^{m}$, i.e., $\mathfrak{P}\left(\left(\mathbf{x},\mathbf{y}\right)\right)=\mathbf{y}$.
Since the set $\mathcal{Q}$ is convex, the set $\mathfrak{P}\left(\mathcal{Q}\right)$
is also convex. The quadratic optimization problem \[
\underset{\mathbf{y}\in\mathfrak{P}\left(\mathcal{Q}\right)}{\min}\left\Vert \hat{\mathbf{Y}}-\mathbf{y}\right\Vert _{2}\]
 has a unique solution. Denote this unique solution by $\tilde{\mathbf{y}}$.
Furthermore, recall our assumption that $\mathbf{\Phi}_{T}$ has full
column rank. For any given $\mathbf{y}\in\mathbb{R}^{m}$, the solution
of \[
\underset{\mathbf{x}\in\mathbb{R}^{\left|T\right|}}{\min}\;\left\Vert \mathbf{y}-\mathbf{\Phi}_{T}\mathbf{x}\right\Vert _{2}\]
 is therefore unique. As a result, there exists a unique $\tilde{\mathbf{x}}\in\mathbb{R}^{\left|T\right|}$
such that $\left(\tilde{\mathbf{x}},\tilde{\mathbf{y}}\right)\in\mathcal{Q}$.
This establishes the existence and uniqueness of the point $\left(\tilde{\mathbf{x}},\tilde{\mathbf{y}}\right)$.

\bibliographystyle{ieeetr}
\bibliography{Quantization}

\newpage{}

\begin{figure}
\includegraphics[scale=0.6]{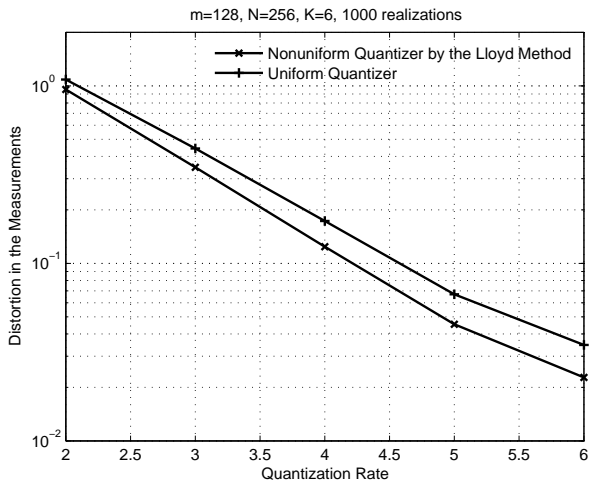}

\caption{\label{fig:dist-measurements}Distortion in the measurements.}

\end{figure}

\begin{figure}
\subfloat[\label{fig:Rec-Dist-Standard}By standard reconstruction algorithms]{\includegraphics[scale=0.6]{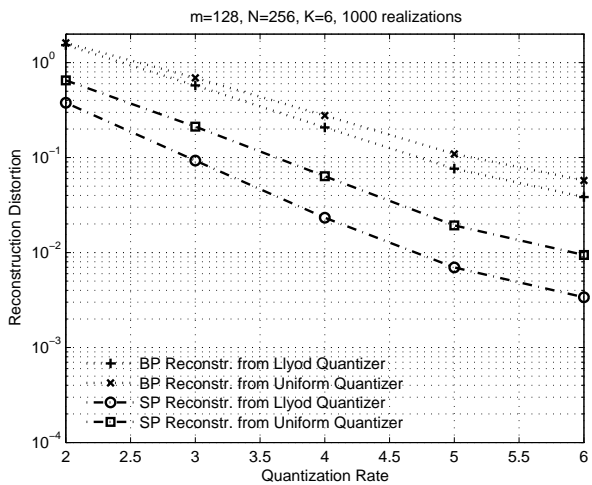}

}

\subfloat[\label{fig:Rec-Dist-Modified}By modified reconstruction algorithms]{\includegraphics[scale=0.6]{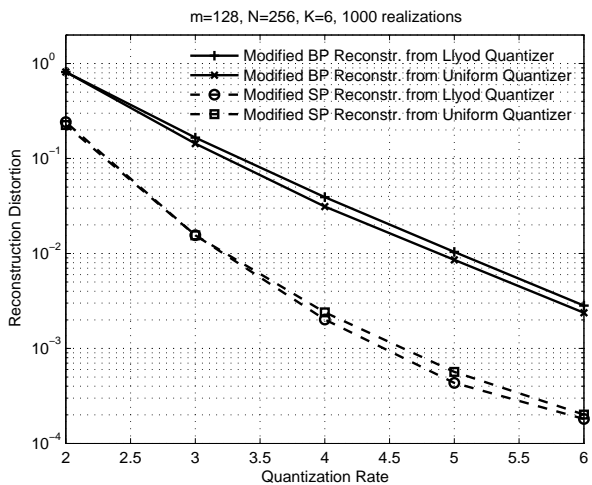}

}

\caption{\label{fig:dist-reconstruction}Distortion in the reconstruction signals. }

\end{figure}

\end{document}